\documentclass{amsart}
\usepackage{amsmath}
\usepackage{amssymb,amsthm}
\usepackage{graphicx}
\numberwithin{equation}{section}
\newtheorem{thm}{Theorem}[section]
\newtheorem{rem}[thm]{Remark}

\newtheorem{lem}[thm]{Lemma}
\newtheorem{cor}[thm]{Corollary}
\newtheorem{dfn}[thm]{Definition}

\subjclass[2010]{Primary 81Q05; Secondary 35P30}
\keywords{Nonlinear eigenvalue problem, Hartree-Fock equations, Critical values}
\title[Hartree-Fock energy functional]{Finiteness of the number of critical values of the Hartree-Fock energy functional less than a constant smaller than the first energy threshold}
\author{Sohei Ashida}
\begin{document}
\maketitle

\begin{abstract}
We study the Hartree-Fock equation and the Hartree-Fock energy functional universally used in many-electron problems. We prove that the set of all critical values of the Hartree-Fock energy functional less than a constant smaller than the first energy threshold is finite. Since the Hartree-Fock equation which is the corresponding Euler-Lagrange equation is a system of nonlinear eigenvalue problems, the spectral theory for linear operators is not applicable. The present result is obtained establishing the finiteness of the critical values associated with orbital energies less than a negative constant and combining the result with the Koopmans' well-known theorem. The main ingredients are the proof of convergence of the solutions and the analysis of the Fr\'echet second derivative of the functional at the limit point.
\end{abstract}

\section{Introduction}\label{firstsec}
In this paper we study the Hartree-Fock equation and the Hartree-Fock energy functional. Fix the number of electrons $N\in \mathbb N$, number of nuclei $n\in \mathbb N$, nuclear charges $Z_j\in \mathbb N,\ j=1,\dots,n$, and the positions of the nuclei $\bar x_j\in \mathbb R^3,\ j=1,\dots,n$. After Hartree \cite{Ha} introduced the Hartree equation ignoring the symmetry with respect to exchanges of variables, considering the symmetry the Hartree-Fock equation was introduced by Fock \cite{Fo} and Slater \cite{Sl} to obtain critical values and the critical points of the functional $\mathcal E(\Phi) = \mathcal E_N(\Phi) := \langle \Psi,H\Psi\rangle$, where
\begin{equation*}\label{myeq1.1}
H:=-\sum_{i=1}^N\Delta_{x_i}+\sum_{i=1}^NV(x_i)+\sum_{1\leq i <j\leq N}\frac{1}{\lvert x_i-x_j\rvert},
\end{equation*}
with $V(x) := -\sum_{j=1}^n\frac{Z_j}{\lvert x-\bar x_j\rvert}$ is an electronic Hamiltonian acting on $L^2(\mathbb R^{3N})$, $\Phi={}^t(\varphi_1,\dots,\varphi_N) \in \bigoplus_{i=1}^N H^1(\mathbb R^3)$ with constraints $\langle \varphi_i,\varphi_j\rangle=\delta_{ij}$, and $\Psi$ is the Slater determinant
\begin{equation*}\label{myeq2.2}
\Psi(x_1,\dots,x_N):=(N!)^{-1/2}\sum_{\tau\in\mathbf S_N}(\mathrm{sgn}\, \tau)\varphi_1(x_{\tau(1)})\cdots\varphi_N(x_{\tau(N)}).
\end{equation*}
Here $\mathbf S_N$ is the symmetric group and $\mathrm{sgn}\, \tau$ is the signature of $\tau$. All functions will be complex valued, but everything in this paper is trivially adapted to spin-dependent functions with only notational changes. The functional $\mathcal E(\Phi)$ can be written explicitly as
\begin{equation}\label{myeq1.2}
\mathcal E(\Phi) =\sum_{i=1}^N\langle \varphi_i ,h\varphi_i\rangle + \frac{1}{2}\int\int\rho(x)\frac{1}{\lvert x-y\rvert}\rho(y)dxdy - \frac{1}{2}\int\int\frac{1}{\lvert x-y\rvert}\lvert\rho(x,y)\rvert^2 dx dy,
\end{equation}
where $h :=-\Delta + V$, $\rho(x) := \sum_{i=1}^N\lvert \varphi_i(x)\rvert^2$ is the density, and
$$\rho(x,y) := \sum_{i=1}^N \varphi_i(x) \varphi_i^*(y),$$
is the density matrix.
The Hartree-Fock equation is the Euler-Lagrange equation corresponding to $\mathcal E(\Phi)$:
\begin{equation}\label{myeq2.1}
(h\varphi_i)(x) + R^{\Phi}(x)\varphi_i(x) - \sum_{j=1}^N Q^{\Phi}_{ij}(x)\varphi_j(x) =
\epsilon_i\varphi_i(x),\ 1 \leq i \leq N,
\end{equation}
with constraints $\langle \varphi_i,\varphi_j\rangle=\delta_{ij}$, where $\epsilon_i\in \mathbb R,\ 1 \leq i \leq N$ are Lagrange multipliers and
\begin{align}
Q^{\Phi}_{ij}(x) &:=\int\lvert x-y\rvert^{-1}\varphi_j^*(y)\varphi_i(y)dy,\notag\\
R^{\Phi}(x) &:=\sum_{i=1}^N\int\lvert x-y\rvert^{-1}\lvert \varphi_i(y)\rvert^2 dy = \sum_{i=1}^NQ_{ii}^{\Phi}(x).\label{myeq3.0.0}
\end{align}
Defining the Fock operator by
$$\mathcal F(\Phi):=h + R^{\Phi} - S^{\Phi},$$
with
\begin{align*}
S^{\Phi} &:= \sum_{i=1}^NS^{\Phi}_{ii},\\
(S^{\Phi}_{ij}w)(x) &:= \left(\int \frac{1}{\lvert x-y\rvert} \varphi_j^*(y) w(y) dy\right)\varphi_i(x),
\end{align*}
the Hartree-Fock equation can also be written as
\begin{equation}\label{myeq1.3}
\mathcal F(\Phi)\varphi_i = \epsilon_i \varphi_i,\ 1 \leq i \leq N.
\end{equation}
[As mentioned in \cite{LS}, the naive Euler-Lagrange equation for $\mathcal E(\Phi)$ is more complicated than the Hartee-Fock equation \eqref{myeq2.1}, but after a unitary change $\varphi_i^{\mathrm{New}}=\sum a_{ij}\varphi_j^{\mathrm{Old}}$, with a unitary $N\times N$ matrix $a_{ij}$, \eqref{myeq2.1} is satisfied by
$$(\varphi_1^{\mathrm{New}},\dots,\varphi_N^{\mathrm{New}}),$$
and some $(\epsilon_1,\dots,\epsilon_N)$. ]

Since the Slater determinant is a sum of products of functions, the Hartree-Fock equation obtained by the method of Lagrange multiplier is a system of nonlinear partial differential equations with unknown constants $\epsilon_i$ which are Lagrange multipliers. We call $(\epsilon_1,\dots,\epsilon_N)$ an orbital energy, if \eqref{myeq2.1} has a solution $(\varphi_1,\dots,\varphi_N)$. One of the difficulty in the analysis of the Hartree-Fock equation is the nonlinearity of the equation. In addition, also the number of constraints of the corresponding minimization problem is a substantial difficulty. In the case of the linear eigenvalue problem of a Hamiltonian, only the norm of a function is constrained. Thus the Lagrange multiplier is equal to a critical value of the functional. However, in the case of the Hartree-Fock equation many Lagrange multipliers concerned with many constraints appear in the equation, and the relation between each multiplier  and the critical value of the functional is not clear. Because of these reasons we can not use the methods for linear eigenvalue problems (spectral theory) to study the critical values of the Hartree-Fock functional and the Lagrange multipliers. For example, we can not see immediately if the critical values or the multipliers are countable or dense in some interval. This problem would be important in the study of the convergence of the approximation methods such as the so-called self-consistent field method. For if the critical values are dense in some interval, it would be hopeless to estimate a chosen one in the interval accurately.

Our main result is that for any $\epsilon > 0$ the set of all critical values less than $J(N-1)-\epsilon$ is finite, where $J(N-1)$ is the lowest critical value for $N-1$ electrons. Since there exists the minimizer of $\mathcal E_{N-1}(\Phi)$ (cf \cite{LS}), $J(N-1)$ is represented also by $J(N-1) =\inf \{\mathcal E_{N-1}(\Phi) :\Phi \in \bigoplus_{i=1}^{N-1}H^1(\mathbb R^3)\}$. This result follows from the finiteness of the set of all critical values of the Hartree-Fock functional with orbital energies $(\epsilon_1 ,\dots ,\epsilon_N)$ satisfying $\epsilon_i < -\epsilon,\ 1 \leq i \leq N$ and the Koopmans' well-known theorem which states that the orbital energies are equal to the ionization potentials. In Lewin \cite{Le} existence of a sequence of critical values less than $J(N-1)$ and converging to $J(N-1)$ has been proved. Thus combining that and the present result we can see that $J(N-1)$ is the lowest limit point of the critical values. When we seek critical values of the Hartree-Fock functional practically, an iterative procedure called self-consistent field (SCF) method is usually adopted to obtain solutions to the Hartree-Fock equation. We obtain as a corollary of the main theorem that the set of all critical values obtained by SCF method with initial functions satisfying a certain condition is finite.

To prove the finiteness of critical values associated with orbital energies less than a negative constant, we first show that if there are infinitely many different critical values, an accumulation point of orbital energies is also an orbital energy and there exists a corresponding sequence of solutions to the Hartree-Fock equation converging to a solution. To show the convergence of the solutions in entire $\mathbb R^3$ we need uniform decay of the solutions. We obtain uniform exponential decay by estimating the uniform decay of $Q_{ij}^{\Phi}(x)$ and $R^{\Phi}(x)$, and using Agmon's method. Next, we investigate the Fr\'echet second derivative of the Hartree-Fock functional at the limit point and show that it is a sum of a positive definite operator and a compact operator. The operators in the Fr\'echet derivative of nonlinear part is decomposed into compact and noncompact operators. We prove that the noncompact operator is positive definite by showing the quadratic form of the operator is written as an integral of a positive function. The linear Hamiltonian part of the Fr\'echet derivative is decomposed into positive definite and compact operators using the resolution of identity. Then since the positive definite operator is an isomorphism, we can apply the theorem by Fu\v cik-Ne\v cas-Sou\v cek-Sou\v cek \cite{FNSS} for real analytic functionals which reduces the problem to that of a finite dimensional real analytic function and applies the result in \cite{SS}. Then it follows that the set of critical values near that corresponding to the accumulation point of the orbital energy is one-point set which contradicts the accumulation of the critical values.

Another fundamental problem of the Hartree-Fock equation is the existence of the solutions which we do not deal with in this paper. Lieb-Simon \cite{LS} proved that if the number of electrons is smaller than or equal to the total charge of the nuclei, there exists a solution to the Hartree-Fock equation which minimizes the functional. Lions \cite{Li} proved that, under the same assumption on the total charge and the number of electrons there exists a sequence of solutions to the Hartree-Fock equation with nonpositive orbital energies and the corresponding critical values are converging to $0$. Lewin \cite{Le} showed that there exist infinitely many critical values of the Hartree-Fock functional less than the first energy threshold $J(N-1)$ and converging to $J(N-1)$ under the same assumption. If we ignore the symmetry and consider the expectation value of the electronic Hamiltonian with respect to a product of functions, the critical points satisfy the Hartree equation instead of the Hartree-Fock equation. For the existence of the solutions to the Hartree equation see e.g. \cite{Re,Wo,St,LS,Li}.

Convergence of the SCF method which is a practical iterative procedure to obtain the solutions to the Hartree-Fock equation is an important problem we do not treat in this paper. A mathematically rigorous  study not depending on the linear combination of atomic orbitals (LCAO) approximation, i.e. a Galerkin approximation which oversimplify the proofs of convergence is in Canc\`es-Bris \cite{CB} which proved that under a certain assumption either the sequence of functions converges to a critical point of the Hartree-Fock functional or it oscillates between two states.

This paper is organized as follows. In Section \ref{secondsec} the main results are stated. Several lemmas needed in the proof of the main result are introduced in Section \ref{thirdsec}. In Section \ref{fourthsec} we prove the main results.

\section{Main results}\label{secondsec}
The main results in this paper are based on the following theorem. For any $\epsilon > 0$ let $\Gamma(\epsilon)$ be the set of all critical values of the Hartree-Fock functional \eqref{myeq1.2} associated with orbital energies satisfying $\epsilon_i < -\epsilon,\ 1 \leq i \leq N$.
\begin{thm}\label{Finthm}
For any $\epsilon > 0$, $\Gamma(\epsilon)$ is finite.
\end{thm}

\begin{rem}
By Theorem \ref{Finthm} we can see that if there exist infinitely many different critical values of $\mathcal E(\Phi)$, there exist infinitely many nonnegative $\epsilon_i$ associated with the critical values or $\epsilon_i$ accumulate at $0$ for some $i$. Theorem \ref{Finthm} does not prohibit existence of infinitely many critical values of $\mathcal E(\Phi)$ less than a negative constant. Actually, it is proved that there exist such critical values in \cite{Le}, and therefore, the condition $\epsilon_i < -\epsilon,\ 1 \leq i \leq N$ is essential for the finiteness.
\end{rem}


Combined with Koopmans' theorem, Theorem \ref{Finthm} yields the following main theorem.
\begin{thm}\label{Jth}
For any $\epsilon>0$ the set of all critical values of the Hartree-Fock functional \eqref{myeq1.2} less than $J(N-1) -\epsilon$ is finite.
\end{thm}

Let us consider the SCF method. In SCF method first, we choose an initial function $\Phi^0 = {}^t(\varphi_1^0 ,\dots ,\varphi_N^0)$.
Next we continue an iterative procedure until the sequence $\{\Phi^j\}$ of the functions obtained in the procedure converges. In the iterative procedure, we find $N$ eigenfunctions $\varphi_1^{j+1} ,\dots ,\varphi_N^{j+1}$ of $\mathcal F(\Phi^j)$ associated with $N$ lowest eigenvalues (including multiplicity) $\mu_1^{j+1} ,\dots ,\mu_N^{j+1}$ and set the next function $\Phi^{j+1} := {}^t(\varphi_1^{j+1} ,\dots ,\varphi_N^{j+1})$. We consider cases in which $\Phi^j$ converges in $\bigoplus_{i=1}^N H^1(\mathbb R^3)$ in the following corollary.
\begin{cor}\label{SCF}
For any $\epsilon>0$ the set of all critical values of the Hartree-Fock functional \eqref{myeq1.2} obtained by SCF method with the initial function $\Phi^0$ satisfying $\mathcal E(\Phi^0) < J(N-1) - \epsilon$ is finite.
\end{cor}

\section{Some preliminaries}\label{thirdsec}
\subsection{Uniform decay and convergence of solutions}
In order to prove Theorem \ref{Finthm} first of all, we need to find an accumulation point of solutions to \eqref{myeq2.1}. To show the convergence of a sequence of solutions in entire $\mathbb R^3$ we need exponential decay of the solutions. Since $R^{\Phi}(x)$ and $Q_{ij}^{\Phi}(x)$ are decaying potentials, we have the exponential decay of solutions basically by the method of Agmon \cite{Ag}.  Because we need uniform exponential decay for a sequence of solutions, we need to estimate the decay of $R^{\Phi}(x)$ and $Q_{ij}^{\Phi}(x)$ uniformly under a weak assumption on the solutions. In the following we use the following standard fact of regularity: if $\Phi \in \bigoplus_{i=1}^N H^1(\mathbb R^3)$ is a solution to the Hartree-Fock equation, then $\Phi \in \bigoplus_{i=1}^N H^2(\mathbb R^3)$ (see e.g. \cite{Li}). We denote the $L^2(\mathbb R^3)$ norm by $\lVert \cdot\rVert$.
\begin{lem}\label{Dbound}
Let $d > 0$ be a constant. Then there exists a constant $C_d >0$ such that any solution $\Phi = {}^t(\varphi_1,\dots,\varphi_N)$ of the Hartree-Fock equation \eqref{myeq2.1} associated with an orbital energy $(\epsilon_1,\dots,\epsilon_N)\in(-d,d)^N$ satisfies
$$\lVert \Delta\varphi_i \rVert < C_d,\ 1 \leq i \leq N.$$
\end{lem}
\begin{proof}
The Hartree-Fock equation \eqref{myeq2.1} for $\Phi$ is written as
\begin{equation}
-\Delta\varphi_i(x)=(\epsilon_i-V(x)-R^{\Phi}(x))\varphi_i(x)+\sum_{j=1}^NQ_{ij}^{\Phi}(x)\varphi_j(x),\ 1 \leq i \leq N.
\end{equation}\label{meq3.0.0.0.0.1}
Here we notice that since the Coulomb potential is $\Delta$-bounded with relative bound $0$ (see e.g. \cite[Chapter V. \S5]{Ka}), $V$ is $\Delta$-bounded with relative bound $0$. Thus for any $0<a<1$ there exists $b>0$ such that
$$\lVert Vu\rVert\leq a\lVert \Delta u\rVert + b \lVert u\rVert.$$
Since the center of the Coulomb potential is irrelevant to the relative bound, for any $0<\tilde a<1$ there exists $\tilde b>0$ such that for any $x\in \mathbb R^3$
$$\lVert \lvert x-y\rvert^{-1}u(y)\rVert_{L^2(\mathbb R^3_y)}\leq\tilde a \lVert \Delta u\rVert+\tilde b\lVert  u\rVert.$$
By the constraints $\lVert\varphi_i\rVert=1,\ 1 \leq i \leq N,$ we obtain
$$\lvert Q_{ij}^{\Phi}(x)\rvert \leq \lVert\lvert x-y\rvert^{-1}\varphi_i(y)\rVert_{L^2(\mathbb R^3_y)}\lVert \varphi_j\rVert \leq \tilde a \lVert\Delta\varphi_i\rVert+\tilde b.$$
By \eqref{myeq3.0.0} we also have
$$\lvert R^{\Phi}(x)\rvert \leq N\tilde a \lVert\Delta\varphi_i\rVert + N\tilde b.$$
Thus by \eqref{meq3.0.0.0.0.1} we find
$$\lVert\Delta\varphi_i\rVert\leq (a+2N\tilde a)\lVert\Delta \varphi_i\rVert+(b+2N\tilde b+d).$$
Since we can choose arbitrarily small $a$ and $\tilde a$, we may suppose $a+2N\tilde a<1$. Hence we obtain
$$\lVert\Delta\varphi_i\rVert\leq (1-a-2N\tilde a)^{-1}(b+2N\tilde b+d),$$
which completes the proof.
\end{proof}

Since $\lvert x\rvert^{-1}$ is $\Delta$-bounded, there exists a constant $C>0$ such that the following inequality holds.
\begin{equation}\label{myeq3.0.0.0.3}
\lVert \lvert x\rvert^{-1}u\rVert\leq C(\lVert\Delta u\rVert+\lVert u\rVert).
\end{equation}
We have the following uniform exponential decay of solutions to the Hartree-Fock equations associated with orbital energies less than a negative constant and satisfying a weak decay condition.
\begin{lem}\label{expdecay}
Let $\epsilon > \tilde \epsilon > 0$, $d ,r_0 > 0$ and $C_d$ be the constant in Lemma \ref{Dbound}. Then there exists $\tilde C > 0$ such that for any solution $\Phi = {}^t(\varphi_1,\dots,\varphi_N)$ of the Hartree-Fock equation \eqref{myeq2.1} associated with an orbital energy $(\epsilon_1,\dots,\epsilon_N)\in(-d,-\epsilon)^N$ and satisfying $\lVert\varphi_i\rVert_{L^2(\mathbb R^3\setminus B_{r_0})}< \frac{\epsilon-\tilde \epsilon}{8NC(C_d+1)},\ 1 \leq i \leq N$ the following estimate holds.
$$\lVert\exp(\tilde \epsilon^{1/2}\lvert x\rvert)\varphi_i(x)\rVert\leq \tilde C,\ 1 \leq i \leq N.$$
\end{lem}
\begin{proof}
Let us first estimate the decay of $Q_{ij}^{\Phi}(x)$. 
Since the center of the Coulomb potential is irrelevant to the relative bound,  by \eqref{myeq3.0.0.0.3} and Lemma \ref{Dbound} the following holds for any $x\in \mathbb R^3$.
\begin{equation}\label{myeq3.0}
\lVert\lvert x-y\rvert^{-1}\varphi_i\rVert_{L^2(\mathbb R^3_y)}\leq C(\lVert \Delta \varphi_i\rVert+\lVert \varphi_i\rVert)\leq C(C_d+1),
\end{equation}
where we used the constraint $\lVert \varphi_i\rVert=1$. We divide $Q_{ij}^{\Phi}(x)$ into two parts:
\begin{align*}
Q_{ij}^{\Phi}(x)=&\int_{\lvert y\rvert<\lvert x\rvert/2}\lvert x-y\rvert^{-1}\varphi_j^*(y)\varphi_i(y)dy\\
&\quad + \int_{\lvert y\rvert\geq\lvert x\rvert/2}\lvert x-y\rvert^{-1}\varphi_j^*(y)\varphi_i(y)dy.
\end{align*}
The first term is estimated as
\begin{align*}
\left\lvert\int_{\lvert y\rvert<\lvert x\rvert/2}\lvert x-y\rvert^{-1}\varphi_j^*(y)\varphi_i(y)dy\right\rvert&\leq2\lvert x\rvert^{-1}\int_{\lvert y\rvert<\lvert x\rvert/2}\lvert\varphi_j^*(y)\varphi_i(y)\rvert dy\\
&\leq2\lvert x\rvert^{-1}\lVert \varphi_j\rVert\lVert \varphi_i\rVert=2\lvert x\rvert^{-1}.
\end{align*}
The second term is estimated as
\begin{align*}
\left\lvert\int_{\lvert y\rvert\geq\lvert x\rvert/2}\lvert x-y\rvert^{-1}\varphi_j^*(y)\varphi_i(y)dy\right\rvert&\leq\lVert\lvert x-y\rvert^{-1}\varphi_i(y)\rVert_{L^2(\mathbb R^3_y)}\lVert \varphi_j\rVert_{L^2(\mathbb R^3\setminus B_{\lvert x\rvert/2})}\\
&\leq C(C_d+1)\lVert \varphi_j\rVert_{L^2(\mathbb R^3\setminus B_{\lvert x\rvert/2})},
\end{align*}
where we used \eqref{myeq3.0} in the second inequality. Thus by the assumption we have
\begin{equation}\label{myeq3.0.1}
\lvert Q_{ij}^{\Phi}(x)\rvert < 2\lvert x\rvert^{-1} + C(C_d+1)\frac{\epsilon-\tilde \epsilon}{8NC(C_d+1)} < \frac{\epsilon-\tilde \epsilon}{4N},
\end{equation}
for $\lvert x\rvert>r_1:=\max\{ 2r_0, \frac{16N}{\epsilon-\tilde \epsilon}\}$.
By \eqref{myeq3.0.0} we also have
\begin{equation}\label{myeq3.0.2}
\lvert R^{\Phi}(x)\rvert< \frac{\epsilon-\tilde \epsilon}{4},
\end{equation}
for $\lvert x\rvert>r_1$.

Let $\eta(r)\in C_0^{\infty}(\mathbb R)$ be a function such that $\eta(r)=r$ for $-1<r<1$ and $\lvert\eta'(r)\rvert\leq 1$.
Set $\rho_k(x) := \tilde\epsilon^{1/2}k\eta(\langle x\rangle/k)$ and $\chi_k(x):=e^{\rho_k(x)}$, where $\langle x\rangle:=\sqrt{1+\lvert x\rvert^2}$.
By a direct calculation we have
\begin{align*}
\mathrm{Re}\, \langle(-\Delta\varphi_i),\chi_k^2\varphi_i\rangle&=\rVert\nabla(\chi_k\varphi_i)\lVert^2-\lVert (\nabla \chi_k)\varphi_i\rVert^2\\
&=\rVert\nabla(\chi_k\varphi_i)\lVert^2-\lVert (\nabla \rho_k)\chi_k\varphi_i\rVert^2.
\end{align*}
Hence by \eqref{myeq2.1} we have
\begin{equation}\label{myeq3.0.2.1}
\begin{split}
0 &= \sum_{i=1}^N \mathrm{Re}\, \langle (-\Delta +V(x) + R^{\Phi}(x)  - \epsilon_i) \varphi_i - \sum_{j=1}^NQ_{ij}^{\Phi}(x)\varphi_j,\chi_k^2 \varphi_i\rangle\\
&= \sum_{i=1}^N\bigg\{\rVert\nabla(\chi_k\varphi_i)\lVert^2-\lVert (\nabla \rho_k)\chi_k\varphi_i\rVert^2 + \langle(V(x)+R^{\Phi}(x) - \epsilon_i)\chi_k\varphi_i,\chi_k\varphi_i\rangle\\
&\qquad  - \sum_{j=1}^N\langle Q_{ij}^{\Phi}(x)\chi_k\varphi_j,\chi_k\varphi_i\rangle\bigg\}\\
&\geq\sum_{i=1}^N\bigg\{\langle(V(x)+R(x)-\epsilon_i-\lvert\nabla\rho_k\rvert^2)\chi_k\varphi_i,\chi_k\varphi_i\rangle\\
&\qquad-\sum_{j=1}^N(\langle2^{-1}\lvert Q_{ij}^{\Phi}(x)\rvert\chi_k\varphi_j,\chi_k\varphi_j\rangle+\langle2^{-1}\lvert Q_{ij}^{\Phi}(x)\rvert\chi_k\varphi_i,\chi_k\varphi_i\rangle)\bigg\}\\
&=\sum_{i=1}^N\langle(U_i(x)-\epsilon_i-\lvert\nabla\rho_k\rvert^2)\chi_k\varphi_i,\chi_k\varphi_i\rangle\\
&\geq \sum_{i=1}^N\langle(U_i(x)-\epsilon_i-\tilde \epsilon)\chi_k\varphi_i,\chi_k\varphi_i\rangle,
\end{split}
\end{equation}
where $U_i(x):=V(x)+R^{\Phi}(x)-\sum_{j=1}^N(\lvert Q_{ij}^{\Phi}(x)\rvert+\lvert Q_{ji}^{\Phi}(x)\rvert)/2$. Here we used 
$$\lvert\langle Q_{ij}^{\Phi}(x)\chi_k\varphi_j,\chi_k\varphi_i\rangle\rvert\leq \frac{1}{2}\{\langle\lvert Q_{ij}^{\Phi}(x)\rvert\chi_k\varphi_i,\chi_k\varphi_i\rangle+\langle\lvert Q_{ij}^{\Phi}(x)\rvert\chi_k\varphi_j,\chi_k\varphi_j\rangle\},$$
in the second step and $\lvert\nabla\rho_k\rvert^2<\tilde \epsilon$ in the last step.

Noticing $\lvert V(x)\rvert\to0$ as $\lvert x\rvert\to\infty$, by \eqref{myeq3.0.1} and \eqref{myeq3.0.2} there exists $r_2>r_1$ independent of $(\varphi_1,\dots,\varphi_N)$ such that $\lvert U_i(x)\rvert<(\epsilon - \tilde\epsilon)/2,\ i=1,\dots,N$ for $\lvert x\rvert >r_2$. Thus
\begin{equation}\label{myeq3.0.2.1.1}
\begin{split}
&\sum_{i=1}^N\int_{\lvert x\rvert>r_2}(-\epsilon_i-\tilde \epsilon-(\epsilon-\tilde\epsilon)/2)\lvert\chi_k\varphi_i(x)\rvert^2dx\\
&\quad\leq\sum_{i=1}^N\int_{\lvert x\rvert>r_2}(-\epsilon_i-\tilde \epsilon+U_i(x))\lvert\chi_k\varphi_i(x)\rvert^2dx\\
&\quad\leq\sum_{i=1}^N\int_{\lvert x\rvert\leq r_2}(\epsilon_i+\tilde \epsilon-U_i(x))\lvert\chi_k\varphi_i(x)\rvert^2dx\\
&\quad\leq\sum_{i=1}^N\sup_{\lvert x\rvert\leq r_2}\lvert\chi_k\rvert^2\int_{\lvert x\rvert\leq r_2}\lvert \epsilon_i+\tilde \epsilon-U_i(x)\rvert\lvert\varphi_i(x)\rvert^2dx,
\end{split}
\end{equation}
where we used \eqref{myeq3.0.2.1} in the second inequality.
Because $\lvert\eta(r)\rvert \leq \lvert r\rvert$, we can estimate $\chi_k$ as $\sup_{\lvert x\rvert\leq r_2}\lvert\chi_k\rvert^2\leq e^{2\tilde\epsilon^{1/2}\sqrt{1+r_2^2}}$.
Since $V$ and $\lvert x-y\rvert^{-1}$ are $\Delta$-bounded, $\lVert\varphi_i\rVert=1$, $\lvert \epsilon_i\rvert \leq d$ and $\lVert\Delta\varphi_i\rVert<C_d$, there exists a constant $C_1>0$ independent of $(\varphi_1,\dots,\varphi_N)$ such that
$$\int_{\lvert x\rvert\leq r_2}\lvert \epsilon_i+\tilde \epsilon-U_i(x)\rvert \lvert\varphi_i(x)\rvert^2dx
\leq C_1.$$
Thus the last expression in \eqref{myeq3.0.2.1.1} is bounded by $Ne^{2\tilde\epsilon^{1/2}\sqrt{1+r_2^2}}C_1$ and we have
\begin{equation}\label{myeq3.0.2.1.2}
\begin{split}
Ne^{2\tilde\epsilon^{1/2}\sqrt{1+r_2^2}}C_1&\geq\sum_{i=1}^N\int_{\lvert x\rvert>r_2}(-\epsilon_i-\tilde \epsilon-(\epsilon-\tilde\epsilon)/2)\lvert\chi_k\varphi_i(x)\rvert^2dx\\
&\geq(\epsilon - \tilde\epsilon)/2\sum_{i=1}^N\int_{\lvert x\rvert>r_2}\lvert\chi_k\varphi_i(x)\rvert^2dx,
\end{split}
\end{equation}
where we used $-\epsilon_i > \epsilon$.
By Fatou's lemma we obtain
\begin{equation}\label{myeq3.0.2.1.3}
\liminf_{k\to\infty}\sum_{i=1}^N\int_{\lvert x\rvert>r_2}\lvert\chi_k\varphi_i(x)\rvert^2dx \geq \sum_{i=1}^N\int_{\lvert x\rvert>r_2}\lvert e^{\tilde \epsilon^{1/2}\langle x\rangle}\varphi_i(x)\rvert^2dx.
\end{equation}
The result of the lemma follows immediately from \eqref{myeq3.0.2.1.2} and \eqref{myeq3.0.2.1.3}.
\end{proof}

Using the uniform exponential decay in Lemma \ref{expdecay} we have the following lemma.

\begin{lem}\label{convseq}
Let $\mathbf e^m:=(\epsilon_1^m,\dots,\epsilon_N^m),\ m=1,2,\dots$ be a sequence of orbital energies converging to $\mathbf e^{\infty}:=(\epsilon_1^{\infty},\dots,\epsilon_N^{\infty})\in(-\infty,0)^N$ and $\Phi^m := {}^t(\varphi_1^m,\dots,\varphi_N^m)$ be the associated solutions to the Hartree-Fock equation \eqref{myeq2.1}. Then $\mathbf e^{\infty}$ is an orbital energy and there exists a subsequence of $\Phi^m$ converging in $\bigoplus_{i=1}^NH^2(\mathbb R^3)$ to a solution of the Hartree-Fock equation associated with $\mathbf e^{\infty}$.
\end{lem}
\begin{proof}
Since $\mathbf e^m$ converges to $\mathbf e^{\infty}:=(\epsilon_1^{\infty},\dots,\epsilon_N^{\infty})\in(-\infty,0)^N$, there exists $d>0$ and $\epsilon>0$ such that $\mathbf e^m\in(-d,-\epsilon)^N$ for any $m$. Thus by Lemma \ref{Dbound} there exists $C >0$ such that $\lVert \Phi^m\rVert_{\bigoplus_{i=1}^NH^2(\mathbb R^3)}<C$ for any $m$. Therefore, by the Rellich selection theorem for any $p\in \mathbb N$ there exists a Cauchy subsequence $\{\varphi_i^{m_l}\}$ of $\{\varphi_i^{m}\}$ in $L^2(B_p)$, where $B_r:=\{x\in \mathbb R^3:\lvert x\rvert<r\}$. The Cauchy sequence $\{\varphi_i^{m_l}\}$ satisfies
$$\lVert\varphi_i^{m_{l_1}}-\varphi_i^{m_{l_2}}\rVert_{L^2(B_{p})} \to 0,$$
 as $l_1,l_2\to\infty$. Thus we can choose further a subsequence $\{\varphi_i^{m_k}\}$ of $\{\varphi_i^m\}$ such that
\begin{equation}\label{myeq3.0.2.2}
\lVert\varphi_i^{m_{k_1}}-\varphi_i^{m_{k_2}}\rVert_{L^2(B_{k_0})} < k_0^{-1},
\end{equation}
where $k_0:=\min\{k_1,k_2\}$. We may assume $\{\varphi_i^k\}$ itself is a sequence satisfying the condition in \eqref{myeq3.0.2.2}. Using the constraint $\lVert\varphi_i^{k}\rVert=1$, we can see that for any $\delta>0$ there exist $r_0>0$ and $l_0\in\mathbb N$ such that $\lVert\varphi_i^{k}\rVert_{L^2(\mathbb R^3\setminus B_{r_0})}=(1-\lVert\varphi_i^{k}\rVert_{L^2(B_{r_0})}^2)^{1/2}<\delta,\ \forall k>l_0$. Accordingly, by Lemma \ref{expdecay} there exists $\tilde C>0$ such that $\lVert \langle x\rangle\varphi_i^{k}\rVert\leq\tilde C$ for any $i$ and sufficiently large $k$.
Since $\lvert x\rvert \geq k$ for $x \in \mathbb R^3\setminus B_{k}$, we have
$$\lVert \varphi_i^{k}\rVert_{L^2(\mathbb R^3\setminus B_{k})}\leq k^{-1}\lVert \langle x\rangle\varphi_i^{k}\rVert \leq \tilde C k^{-1}.$$
Therefore, we obtain
\begin{align*}
\lVert\varphi_i^{k_1}-\varphi_i^{k_2}\rVert_{L^2(\mathbb R^3)} &\leq \lVert\varphi_i^{k_1}-\varphi_i^{k_2}\rVert_{L^2(B_{k_0})} + \lVert\varphi_i^{k_1}-\varphi_i^{k_2}\rVert_{L^2(\mathbb R^3\setminus B_{k_0})}\\
&\leq k_0^{-1} + 2\tilde Ck_0^{-1}.
\end{align*}
for sufficiently large $k_0$. Thus $\{\varphi_i^{k}\}$ is a Cauchy sequence in $L^2(\mathbb R^3)$.

Set $Q_{ij}^k(x) := Q_{ij}^{\Phi^k}(x)$ and $R^k(x) := R^{\Phi^k}(x)$. Then by \eqref{myeq3.0.0.0.3} there exists a constant $C_1 >0$ such that
\begin{align*}
\lvert Q^{k_1}_{ij}(x) - Q^{k_2}_{ij}(x)\rvert &\leq \Bigg\lvert\int\lvert x-y\rvert^{-1}((\varphi_j^{k_1})^*(y)-(\varphi_j^{k_2})^*(y))\varphi_i^{k_1}(y)dy\\
&\quad +\int\lvert x-y\rvert^{-1}(\varphi_j^{k_2})^*(y)(\varphi_i^{k_1}(y)-\varphi_i^{k_2}(y))dy\Bigg\rvert\\
&\leq C_1\sum_{l=1}^N\lVert\varphi_l^{k_1}-\varphi_l^{k_2}\rVert (\lVert \varphi_i^{k_1}\rVert_{H^2(\mathbb R^3)} +\lVert \varphi_j^{k_2}\rVert_{H^2(\mathbb R^3)})\\
&\leq 2CC_1\sum_{l=1}^N\lVert\varphi_l^{k_1}-\varphi_l^{k_2}\rVert.
\end{align*}
By \eqref{myeq3.0.0} we also have
$$\lvert R^{k_1}(x) -R^{k_2}(x)\rvert \leq 2NCC_1\sum_{l=1}^N\lVert\varphi_l^{k_1}-\varphi_l^{k_2}\rVert.$$
Moreover we can easily see that there exists a constant $C_2 > 0$ such that
$$\lvert Q^{k}_{ij}(x)\rvert, \lvert R^{k}(x)\rvert <C_2,$$
for any $k$.

Thus using the Hartree-Fock equation \eqref{myeq2.1} we can see that there exists $C_3>0$ such that
\begin{equation}\label{myeq3.0.2.3}
\begin{split}
&\lVert h(\varphi_i^{k_1}-\varphi_i^{k_2})\rVert\\
&\quad=\left\lVert (\epsilon_i^{k_1}-R^{k_1}(x))\varphi_i^{k_1}+\sum_{j=1}^NQ^{k_1}_{ij}(x)\varphi_j^{k_1}-(\epsilon_i^{k_2}-R^{k_2}(x))\varphi_i^{k_2} - \sum_{j=1}^NQ^{k_2}_{ij}(x)\varphi_j^{k_2}\right\rVert\\
&\quad\leq C_3\sum_{l=1}^N \lVert \varphi_l^{k_1}-\varphi_l^{k_2}\rVert + \lvert\epsilon^{k_1}_i-\epsilon^{k_2}_i\rvert.
\end{split}
\end{equation}
Because $V$ is $\Delta$-bounded with relative bound smaller than $1$, $\Delta$ is $h$-bounded. Therefore, by \eqref{myeq3.0.2.3} we can see that $\{\varphi_i^{k}\}$ is a Cauchy sequence in $H^2(\mathbb R^3)$. Let $\varphi_i^{\infty} \in H^2(\mathbb R^3)$ be the limit. Then setting $\Phi^{\infty} := (\varphi_1^{\infty} ,\dots,\varphi_N^{\infty})$ the both sides of the Hartree-Fock equation converge in $L^2(\mathbb R^3)$ and we have
$$h\varphi_i^{\infty} +R^{\Phi^{\infty}}\varphi_i^{\infty} - \sum_{j=1}^N Q_{ij}^{\Phi^{\infty}} \varphi_j^{\infty} = \epsilon_i^{\infty}\varphi_i^{\infty}.$$
Thus $\mathbf e^{\infty}$ is an orbital energy associated with $\Phi^{\infty}$.
\end{proof}

\subsection{Real-analytic operators in Banach space}
In this subsection following \cite{FNSS} we introduce the real-analytic operators and their property. Let $X$ and $Y$ be real Banach spaces.  We denote the norm of $X$ by $\lVert \cdot\rVert$.

\begin{dfn}
Let $D$ be an open subset of $X$. The mapping $F:D\to Y$ is said to be real-analytic on $D$ if the following conditions are fulfilled:
\begin{itemize}
\item[(i)] For each $x\in D$ there exist Fr\'echet derivatives of arbitrary orders $d^mF(x,\dots)$.
\item[(ii)] For each $x\in D$ there exists $\delta>0$ such that for any $h\in X$ satisfying $\lVert h\rVert<\delta$ one has
$$F(x+h)=\sum_{m=0}^{\infty}\frac{1}{m!}d^mF(x,h^m),$$
(the convergence being locally uniform and absolute), where $h^m := [h, \dots ,h]$ ($m$-times).
\end{itemize}
\end{dfn}

\begin{lem}[{\cite[Theorem 4.1]{FNSS}}]\label{onepoint}
Let $f$ be a real-analytic functional on  a Banach space $Y_1$ and let $Y_2$ be another Banach space. Suppose that there exists a bilinear form $\langle\langle\cdot,\cdot\rangle\rangle$ on $Y_1\times Y_2$ such that for fixed $y\in Y_1$, $\langle\langle y,\cdot\rangle\rangle$ is continuous on $Y_2$ and $\langle\langle y,x\rangle\rangle=0$ for all $y\in Y_1$ implies $x=0$. For each $y\in Y_1$ suppose there exists $F(y)$ such that
\begin{itemize}
\item[(f1)] $df(y,h)=\langle\langle h,F(y)\rangle\rangle$ for each $h\in Y_1$.
\end{itemize}
Let the operator
\begin{itemize}
\item[(f2)] $F: Y_1\to Y_2$ is real-analytic.
\end{itemize}
Denote $\mathcal B_f:=\{y\in Y_1: f'(y)=0\}$ and let $y_0\in \mathcal B_f$. Suppose that
\begin{itemize}
\item[(f3)] $$F'(y_0)=L+M,$$
where $L$ is an isomorphism of $Y_1$ onto $Y_2$ and $M$ is a compact operator.
\end{itemize}
Then  there exists a neighborhood $U(y_0)$ in $Y_1$ of a point $y_0$ such that $f(\mathcal B_f\cap U(y_0))$ is a one-point set.
\end{lem}

\section{Proofs of Theorems}\label{fourthsec}
In this section we prove main results. The method is based on analysis of Fr\'echet second derivative of the Hartree-Fock functional at the limit point of solutions to the Hartree-Fock equation.
\begin{proof}[Proof of Theorem \ref{Finthm}]
We assume there are infinitely many critical values $\mathcal E(\Phi^m)$ associated with orbital energies
$$\mathbf e^m=(\epsilon^m_1,\dots,\epsilon^m_N)\in(-\infty,-\epsilon)^N,$$
and critical points
$$\Phi^m={}^t(\varphi^m_1,\dots,\varphi^m_N),$$
and they satisfy $\mathcal E(\Phi^{m_1})\neq \mathcal E(\Phi^{m_2}),\ m_1\neq m_2$.
We shall show that this assumption leads to a contradiction.

\vspace{\baselineskip}

\noindent \textit{Step 1}
For $\Phi={}^t(\varphi_1,\dots,\varphi_N)\in \bigoplus_{i=1}^NH^2(\mathbb R^3)$ we shall show
\begin{equation}\label{myeq4.3}
R^{\Phi}-S^{\Phi}\geq0.
\end{equation}
Since $R^{\Phi}-S^{\Phi} =\sum_{i=1}^N(Q_{ii}^{\Phi}-S_{ii}^{\Phi})$, for this purpose we have only to show
$$Q_{ii}^{\Phi}-S_{ii}^{\Phi} \geq 0,\ 1 \leq i \leq N.$$
Let $w \in L^2(\mathbb R^3)$. We define 
$$\hat \Psi_i := 2^{-1/2}(w(x)\varphi_i(y)-\varphi_i(x)w(y)).$$
Then we can easily check that
$$\langle w ,(Q_{ii}^{\Phi}-S_{ii}^{\Phi})w\rangle = \int \frac{1}{\lvert x-y\rvert}\lvert \hat \Psi_i\rvert^2 dx dy \geq 0.$$
Thus we have $Q_{ii}^{\Phi}-S_{ii}^{\Phi} \geq 0$.

Multiplying $\varphi_i^*$ to the Hartree-Fock equation \eqref{myeq1.3} and integrating the both sides we obtain by \eqref{myeq4.3}
$$\epsilon_i=\langle\varphi_i,h\varphi_i\rangle+\langle\varphi_i,R^{\Phi}\varphi_i\rangle-\langle\varphi_i,S^{\Phi}\varphi_i\rangle \geq \langle\varphi_i,h\varphi_i\rangle \geq \inf\sigma(h)>-\infty,$$
where $\sigma(h)$ is the spectra of $h$. Thus, if there are infinitely many orbital energies $\mathbf e^m=(\epsilon_1^m,\dots,\epsilon_N^m)\in (-\infty,-\epsilon)^N,\ m=1,2,\dots$, we have $\mathbf e^m\in (\inf\sigma(h),-\epsilon)^N$ and thus there exists a subsequence $\{\mathbf e^{m_k}\}$ such that $\epsilon_i^{m_k}$ converges to $\epsilon_i^{\infty}\leq -\epsilon$. By Lemma \ref{convseq}, $\mathbf e^{\infty}:=(\epsilon_1^{\infty},\dots,\epsilon_N^{\infty})$ is an orbital energy. Denoting the subsequence $\{\mathbf e^{m_k}\}$ again by $\{\mathbf e^m\}$, we may assume $\mathbf e^m$ converges to $\mathbf e^{\infty}$, and taking a subsequence further the associated solution $\Phi^m ={}^t(\varphi_1^m,\dots,\varphi_N^m)$  to the Hartree-Fock equation converges in $\bigoplus_{i=1}^NH^2(\mathbb R^3)$ to a solution $\Phi^{\infty} ={}^t(\varphi_1^{\infty},\dots,\varphi_N^{\infty})$ associated with $\mathbf e^{\infty}$.

\vspace{\baselineskip}

\noindent \textit{Step 2}. Let us define
\begin{equation*}
R^{\Phi}_i(x):=\sum_{j\neq i}\int\lvert x-y\rvert^{-1}\varphi_j^*(y)\varphi_j(y) dy,
\end{equation*}
\begin{equation*}
S^{\Phi}_i:=\sum_{j\neq i} S_{jj}^{\Phi}.
\end{equation*}
Then by \eqref{myeq3.0.0} and $Q_{ii}^{\Phi}\varphi_i=S_{ii}^{\Phi}\varphi_i$, the Hartree-Fock equation \eqref{myeq1.3} is written as
\begin{equation}\label{myeq4.4}
h\varphi_i+R_i^{\Phi}\varphi_i-S_i^{\Phi}\varphi_i-\epsilon_i\varphi_i=0,
\end{equation}
Denote by $Y_1:=(\bigoplus _{i=1}^NH^2(\mathbb R^3))\bigoplus \mathbb R^N$ and $Y_2:=(\bigoplus _{i=1}^NL^2(\mathbb R^3))\bigoplus \mathbb R^N$ the direct sum of Banach spaces regarding $H^2(\mathbb R^3)$ and  $L^2(\mathbb R^3)$ as real Banach spaces with respect to multiplication by real numbers. We define a functional $f:Y_1\to \mathbb R$ by
$$f(\Phi,\mathbf e):=\mathcal E(\Phi)-\sum_{i=1}^N\epsilon_i(\lVert \varphi_i\rVert^2-1),$$
and a bilinear form $\langle\langle\cdot,\cdot\rangle\rangle$ on $Y_1$ and $Y_2$ by
$$\langle\langle[\Phi^1,\mathbf e^1],[\Phi^2,\mathbf e^2]\rangle\rangle:=\sum_{i=1}^N2\mathrm{Re}\, \langle \varphi_i^1,\varphi_i^2\rangle+\sum_{i=1}^N\epsilon_i^1\epsilon_i^2,$$
for $[\Phi^j,\mathbf e^j] \in Y_j,\ j=1,2$.
We also define a mapping $F: Y_1\to Y_2$ by
$$F(\Phi,\mathbf e):=[{}^t(F_1(\Phi,\mathbf e),\dots,F_N(\Phi,\mathbf e)),(1 - \lVert \varphi_1\rVert^2,\dots,1 - \lVert \varphi_N\rVert^2)],$$
$$F_i(\Phi,\mathbf e):=h\varphi_i+R_i^{\Phi}\varphi_i-S_i^{\Phi}\varphi_i-\epsilon_i\varphi_i.$$
Then by \eqref{myeq4.4} we can see that if $\Phi$ is a solution of the Hartree-Fock equation associated with $\mathbf e$, $F(\Phi,\mathbf e)=0$ holds. Moreover, we have
$$df([\Phi,\mathbf e],[\tilde \Phi,\tilde{\mathbf e}])=\langle\langle[\tilde \Phi,\tilde{\mathbf e}],F(\Phi,\mathbf e)\rangle\rangle.$$
Hence $f$ and $F$ satisfy the assumption (f1) in Lemma \ref{onepoint}. It is easily seen that the assumption (f2) is also satisfied. Moreover, we can see that if $F(\Phi,\mathbf e)=0$, then $f(\Phi,\mathbf e)=\mathcal E(\Phi)$. Thus the solutions to the Hartree-Fock equation are critical points of $f$ and the corresponding critical values of $f$ are those of the Hartree-Fock functional.

Therefore, if we prove the assumption (f3) is satisfied at $[\Phi^{\infty},\mathbf e^{\infty}]$, then by Lemma \ref{onepoint} there exists a neighborhood $U$ of $[\Phi^{\infty},\mathbf e^{\infty}]$ such that $f(\mathcal B_{f}\cap U)$ is one-point set. Remembering it has been assumed that the critical values $f(\Phi^m)=\mathcal E(\Phi^m)$ satisfy $f(\Phi^{m_1})\neq f(\Phi^{m_2})$ for $m_1\neq m_2$ from the beginning of the proof, and noticing $f(\Phi^m)\to f(\Phi^{\infty})$ we have a contradiction, which concludes the proof of Theorem \ref{Finthm}.

\vspace{\baselineskip}

\noindent \textit{Step 3}. It remains to show (f3). Let us consider the Fr\'echet derivative of $\tilde F: \bigoplus_{i=1}^NH^2(\mathbb R^3)\to\bigoplus_{i=1}^NL^2(\mathbb R^3)$ defined by
$$\tilde F(\Phi) := {}^t(F_1(\Phi,\mathbf e^{\infty}),\dots,F_N(\Phi,\mathbf e^{\infty})).$$
For a mapping $G: \bigoplus_{i=1}^NH^2(\mathbb R^3)\to L^2(\mathbb R^3)$ and $w_j \in H^2(\mathbb R^3)$ we set 
$$G'_{j}w_j:=\lim_{t\to0}[G(\varphi_1^{\infty},\dots,\varphi_j^{\infty}+tw_j,\dots,\varphi_N^{\infty})-G(\varphi_1^{\infty},\dots,\varphi_j^{\infty},\dots,\varphi_N^{\infty})]/t.$$
Then
\begin{align*}
&[R_i^{\Phi}\varphi_i]'_{i}w_i=R_i^{\Phi^{\infty}}w_i,\\
&[R_i^{\Phi}\varphi_i]'_{j}w_j=S_{ij}^{\Phi^{\infty}}w_j+\bar S_{ij}^{\Phi^{\infty}}w_j,\ j\neq i,
\end{align*}
where
$$(\bar S_{ij}^{\Phi}w)(x):=\left(\int\lvert x-y\rvert^{-1}\varphi_j(y)w^*(y)dy\right)\varphi_i(x),$$
and
\begin{align*}
&[S_i^{\Phi}\varphi_i]'_{i}w_i=S_i^{\Phi^{\infty}}w_i,\\
&[S_i^{\Phi}\varphi_i]'_{j}w_j=Q_{ij}^{\Phi^{\infty}}w_j+\bar S_{ji}^{\Phi^{\infty}}w_j,\ j\neq i.
\end{align*}

Set $W:={}^t(w_1,\dots,w_N)$. We define mappings
$$\mathcal R,\mathcal Q: \bigoplus_{i=1}^NH^2(\mathbb R^3)\to\bigoplus_{i=1}^NL^2(\mathbb R^3),$$
by
\begin{align*}
(\mathcal RW)_i&:=R_i^{\Phi^{\infty}}w_i,\\
(\mathcal QW)_i&:=\sum_{j\neq i}Q_{ij}^{\Phi^{\infty}}w_j.
\end{align*}
We shall show that $\mathcal R-\mathcal Q$ is a positive definite operator as an operator on the Hilbert space $\bigoplus_{i=1}^NL^2(\mathbb R^3)$. 
For this purpose we introduce a function
\begin{align*}
&\tilde\Psi(x_1,\dots,x_{N})\\
&\quad:=2^{-1/2}\sum_{i=1}^N\sum_{j\neq i}(w_i(x_1)\varphi_j^{\infty}(x_2)-w_j(x_1)\varphi_i^{\infty}(x_2))\check \Psi_{ij}(x_3,\dots,x_N),
\end{align*}
where
\begin{align*}
&\check \Psi_{ij}(x_3,\dots,x_N):=\varphi_1^{\infty}(x_{\kappa_{ij}(1)})\cdots\varphi_{i-1}^{\infty}(x_{\kappa_{ij}(i-1)})\varphi_{i+1}^{\infty}(x_{\kappa_{ij}(i+1)})\cdot\\
&\hspace{50pt}\cdots\varphi_{j-1}^{\infty}(x_{\kappa_{ij}(j-1)})\varphi_{j+1}^{\infty}(x_{\kappa_{ij}(j+1)})\cdots\varphi_N^{\infty}(x_{\kappa_{ij}(N)}).
\end{align*}
Here $\kappa_{ij}$ is an arbitrarily chosen map from $\{1,\dots, \check i,\dots,\check j,\dots,N\}$ onto $\{3,\dots,N\}$ where $\check i$ means that $i$ is excluded.
We use the notation $[\tilde i\tilde j\vert kl]$ defined by
$$[\tilde i\tilde j\vert  k l]:=\int \lvert x-y\rvert^{-1}w_i^*(x)w_j(x)(\varphi_k^{\infty})^*(y)\varphi_l^{\infty}(y)dxdy.$$
Then we can calculate as
$$\langle W,(\mathcal R-\mathcal Q)W\rangle=\sum_{i=1}^N\sum_{j\neq i}\{[\tilde i\tilde i\vert jj]-[\tilde i\tilde j\vert ji]\}.$$
On the other hand using the constraints $\langle \varphi_i^{\infty} ,\varphi_j^{\infty}\rangle = \delta_{ij}$ we can calculate as
\begin{align*}
&\int dx_1\cdots dx_N\lvert x_1-x_2\rvert^{-1}\lvert \tilde \Psi(x_1,\dots,x_{N})\rvert^2\\
&\quad=2^{-1}\sum_{i=1}^N\sum_{j\neq i}\int dx_1dx_2\lvert x_1-x_2\rvert^{-1}\lvert w_i(x_1)\varphi_j^{\infty}(x_2)-w_j(x_1)\varphi_i^{\infty}(x_2)\rvert^2\\
&\quad=2^{-1}\sum_{i=1}^N\sum_{j\neq i}\{[\tilde i\tilde i\vert jj]+[\tilde j\tilde j\vert ii]-[\tilde i\tilde j\vert ji]-[\tilde j\tilde i\vert ij]\}\\
&\quad=\sum_{i=1}^N\sum_{j\neq i}\{[\tilde i\tilde i\vert jj]-[\tilde i\tilde j\vert ji]\}\\
&\quad=\langle W,(\mathcal R-\mathcal Q)W\rangle.
\end{align*}
Since the left-hand side is positive, we can see that $\langle W,(\mathcal R-\mathcal Q)W\rangle\geq0$.

Next we consider $h-\epsilon_i^{\infty}$. Denote the resolution of identity of $h$ by $E(\lambda)$. Then we can decompose $h$ as
$$h=hE(-\epsilon/2)+h(1-E(-\epsilon/2)).$$
Since $\inf \sigma_{ess}(h)=0$, $hE(-\epsilon/2)$ is a compact operator, where $\sigma_{ess}(h)$ is the essential spectra of $h$. Moreover, we have an inequality $h(1-E(-\epsilon/2))\geq -\epsilon/2$ of operators. Since $\epsilon_i^{\infty}\leq-\epsilon$, we obtain $h(1-E(-\epsilon/2))-\epsilon_i^{\infty} \geq \epsilon/2$. Therefore, the operator $\mathcal H: \bigoplus_{i=1}^NH^2(\mathbb R^3)\to\bigoplus_{i=1}^NL^2(\mathbb R^3)$ defined by
$$(\mathcal HW)_i:=(h-\epsilon_i^{\infty})w_i,$$
is decomposed as a sum $\mathcal H=\mathcal H_1+\mathcal H_2$ of positive definite operator 
$$\mathcal H_1 := \mathrm{diag}\, (h(1-E(-\epsilon/2))-\epsilon_1^{\infty} ,\dots ,h(1-E(-\epsilon/2))-\epsilon_N^{\infty}) \geq \epsilon/2,$$
 and a compact operator
 $$\mathcal H_2 := \mathrm{diag}\, (hE(-\epsilon/2) ,\dots ,hE(-\epsilon/2)),$$
 where $\mathrm{diag}\, (A_1 ,\dots ,A_N)$ is the diagonal matrix whose diagonal elements are \newline $A_1 ,\dots ,A_N$.
If we also define $\mathcal S$ and $\bar{\mathcal S}$ by
$$(\mathcal SW)_i=\sum_{j\neq 1}S_{ij}^{\Phi^{\infty}}w_j -S_i^{\Phi^{\infty}}w_i,\ 
(\bar{\mathcal S}W)_i=\sum_{j\neq 1}\bar S_{ij}^{\Phi^{\infty}}w_j,$$
the Frechet derivative of $\tilde F$ at $\Phi^{\infty}$ is
$$\tilde F'(\Phi^{\infty})=\mathcal H_1+\mathcal H_2+\mathcal R-\mathcal Q+\mathcal S+\bar{\mathcal S}-{}^t\bar{\mathcal S}=\mathcal L+\mathcal M,$$
where $\mathcal L:=\mathcal H_1+\mathcal R-\mathcal Q$ and $\mathcal M:=\mathcal H_2+\mathcal S+\bar{\mathcal S}-{}^t\bar{\mathcal S}$.
Since $\mathcal R-\mathcal Q$ is positive definite and $\mathcal H_1 \geq \epsilon/2$, we have $\mathcal L \geq \epsilon/2$ and $\mathcal L$ is invertible. Moreover, since $S_{ij}^{\Phi^{\infty}}$ and $\bar S_{ij}^{\Phi^{\infty}}$ are compact operators, $\mathcal S+\bar{\mathcal S}+{}^t\bar{\mathcal S}$ is compact and $\mathcal M$ is also compact.

\vspace{\baselineskip}

\noindent \textit{Step 4}. Setting
$$\hat F(\mathbf e) := {}^t(F_1(\Phi^{\infty},\mathbf e),\dots,F_N(\Phi^{\infty},\mathbf e)),$$
we can see that
\begin{align*}
&F'(\Phi^{\infty},\mathbf e^{\infty})[\Phi,\mathbf e]\\
&\quad = [\tilde F'(\Phi^{\infty})\Phi+\hat F'(\mathbf e^{\infty})\mathbf e,- 2\mathrm{Re}\, \langle\varphi_1,\varphi_1^{\infty}\rangle,\dots, - 2\mathrm{Re}\, \langle\varphi_N,\varphi_N^{\infty}\rangle]\\
&\quad =L[\Phi,\mathbf e]+M[\Phi,\mathbf e],
\end{align*}
where
\begin{align*}
L[\Phi,\mathbf e]&:=[\mathcal L\Phi,\mathbf e],\\
M[\Phi,\mathbf e]&:=[\mathcal M\Phi-\mathbf e\Phi^{\infty},- 2\mathrm{Re}\, \langle\varphi_1,\varphi_1^{\infty}\rangle-\epsilon_1,\dots, - 2\mathrm{Re}\, \langle\varphi_N,\varphi_N^{\infty}\rangle-\epsilon_N],
\end{align*}
and $\mathbf e\Phi^{\infty} := {}^t(\epsilon_1\varphi_1^{\infty},\dots,\epsilon_N\varphi_N^{\infty})$. We can easily see that $M$ is a compact operator and $L$ is an isomorphism, and therefore, the assumption (f3) of Lemma \ref{onepoint} for $F$ at $[\Phi^{\infty},\mathbf e^{\infty}]$ is satisfied. This completes the proof.
\end{proof}

\begin{proof}[Proof of Theorem \ref{Jth}]
Let $\Phi = {}^t(\varphi_1 ,\dots ,\varphi_N)$  be a critical point of the Hartree-Fock functional. Without loss of generality, we can assume $\epsilon_N = \max\{\epsilon_1,\dots,\epsilon_N\}$. The Hartree-Fock functional can be written as
$$\mathcal E_N(\Phi) = \sum_{j=1}^{N} \langle\varphi_j ,h \varphi_j\rangle + \sum_{1\leq i<j\leq N} (J_{ij}-K_{ij}),$$
where
\begin{align*}
J_{ij} &:= \int \int \lvert \varphi_i(x)\rvert^2 \frac{1}{\lvert x-y\rvert}\lvert \varphi_j(y)\rvert^2 dx dy,\\
K_{ij} &:= \int \int \varphi_i^*(x)\varphi_j(x)\frac{1}{\lvert x-y\rvert}\varphi_j^*(y)\varphi_i(y) dx dy.
\end{align*}
A direct calculation yields
$$\epsilon_N = \langle \varphi_N ,\mathcal F(\Phi)\varphi_N\rangle = \langle\varphi_N ,h \varphi_N\rangle + \sum_{i=1}^{N-1} (J_{iN}-K_{iN}).$$
Thus setting $\hat \Phi := {}^t(\varphi_1 ,\dots ,\varphi_{N-1})$ we can see that
\begin{align*}
\mathcal E_N(\Phi) &= \sum_{j=1}^{N-1} \langle\varphi_j ,h \varphi_j\rangle + \sum_{1\leq i<j\leq N-1} (J_{ij}-K_{ij}) \\
&\quad + \langle\varphi_N ,h \varphi_N\rangle + \sum_{i=1}^{N-1} (J_{iN}-K_{iN})\\
&= \mathcal E_{N-1}(\hat \Phi) + \epsilon_N\\
&\geq J(N-1) + \epsilon_N.
\end{align*}
Hence if $\mathcal E_N(\Phi) < J(N-1) -\epsilon$, then $\epsilon_N < -\epsilon$. Therefore, by Theorem \ref{Finthm} we can see that the set of all critical values $\mathcal E_N(\Phi)$ satisfying $\mathcal E_N(\Phi) < J(N-1) - \epsilon$ is finite.
\end{proof}

\begin{rem}
The equality $\mathcal E_N(\Phi) = \mathcal E_{N-1}(\hat\Phi) + \epsilon_N$ in the proof of Theorem \ref{Jth} above is Koopmans' theorem in which the difference $\mathcal E_{N-1}(\hat\Phi)-\mathcal E_N(\Phi)$ is regarded as the ionization potential, 
\end{rem}

\begin{proof}[Proof of Corollary \ref{SCF}]
For $\Phi = {}^t(\varphi_1 ,\dots ,\varphi_N)$ and $\tilde\Phi = {}^t(\tilde\varphi_1 ,\dots ,\tilde\varphi_N)$ as in \cite{CB} let us introduce the functional
\begin{align*}
\mathcal E(\Phi ,\tilde \Phi)& := \sum_{i=1}^N \langle\varphi_i ,h\varphi_i\rangle +\sum_{i=1}^N \langle \tilde\varphi_i ,\mathcal F(\Phi)\tilde\varphi_i\rangle\\
&=\sum_{i=1}^N \langle\varphi_i ,h\varphi_i\rangle + \sum_{i=1}^N \langle\tilde \varphi_i ,h\tilde \varphi_i\rangle\\
&\quad + \int\int\rho(x)\frac{1}{\lvert x-y\rvert}\tilde \rho(y)dxdy - \int\int\frac{1}{\lvert x-y\rvert}\rho^*(x,y) \tilde\rho(x,y) dx dy,
\end{align*}
where $\tilde \rho(x) := \sum_{i=1}^N \lvert\tilde\varphi_i(x)\rvert^2$ and $\tilde \rho(x,y) := \sum_{i=1}^N \tilde\varphi_i(x) \tilde \varphi^*_i(y)$.
Then we can easily see that  $\mathcal E(\Phi ,\tilde \Phi)$ is symmetric and $\mathcal E(\Phi ,\Phi)=2\mathcal E(\Phi)$. Moreover it has been proved in the proof of \cite[Theorem 7]{CB} that $\mathcal E(\Phi^j ,\Phi^{j+1})$ is decreasing and $\mathcal E(\Phi^j ,\Phi^{j+1}) \leq \mathcal E(\Phi^j ,\Phi^j)$, where $\Phi^j$ is the function obtained in the iterative procedure of SCF method. Therefore, if $\Phi^j$ converges to a critical point $\Phi^{\infty}$, we have
$$2\mathcal E(\Phi^0)=\mathcal E(\Phi^0,\Phi^0) \geq \mathcal E(\Phi^j ,\Phi^{j+1}) \to \mathcal E(\Phi^{\infty},\Phi^{\infty})=2\mathcal E(\Phi^{\infty}).$$
Thus if $\mathcal E(\Phi^0) < J(N-1) - \epsilon$, then $\mathcal E(\Phi^{\infty}) < J(N-1) - \epsilon$. Hence by Theorem \ref{Jth} it follows that the critical value $\mathcal E(\Phi^{\infty})$ must belong to the set of the finite number of critical values.
\end{proof}


\begin{thebibliography}{99}
\bibitem{Ag} S. Agmon,
Lectures on exponential decay of solutions of second-order elliptic equations: Bounds on eigenfunctions of $N$-body Schr\"odinger operations,
Princeton University Press, (1982)
\bibitem{CB} E. Canc\`es, and C. L. Bris,
On the convergence of SCF algorithms for the Hartree-Fock equations,
{\it ESAIM: M2AN} {\bf 34} (2000), 749--774.
\bibitem{Fo} V. Fock,
N\"aherungsmethode zur L\" osing des quantenmechanischen Mehrk\"orperproblems,
{\it Z. phys.} {\bf 61} (1930), 126--148.
\bibitem{FNSS} S. Fu\v cik, J. Ne\v cas, J. Sou\v cek, and V. Sou\v cek,
Upper bound for the number of critical levels for nonlinear operators in Banach spaces of the type of second order nonlinear partial differential operators,
{\it J. Funct. Anal.} {\bf 11} (1972), 314--333.
\bibitem{Ha} D. Hartree,
The wave mechanics of an atom with a non-coulomb central field. Part I. Theory and methods,
{\it Proc. Comb. Phil. Soc.} {\bf 24} (1928), 89--132.
\bibitem{Ka} T. Kato,
Perturbation theory for linear operators,
 Berlin-Heidelberg-New York, Springer, (1976)
 \bibitem{Le} M. Lewin,
Existence of Hatree-Fock excited states for atoms and molecules,
{\it Lett. Math. Phys.} {\bf 108} (2018), 985--1006.
\bibitem{LS} E. H. Lieb and B. Simon,
The Hartree-Fock theory for Coulomb systems,
{\it Commun. math. Phys.} {\bf 53} (1977), 185--194.
\bibitem{Li} P. L. Lions,
Solutions of Hartree-Fock equations for Coulomb systems,
{\it Commun. math. Phys.} {\bf 109} (1987), 33--97.
\bibitem{Re} M. Reeken,
General theorem on bifurcation and its applications to the Hartree equation of Helium atom,
{\it J. Math. Phys.} {\bf 11} (1970), 2505--2512.
\bibitem{Sl} J. C. Slater,
A note on Hartree's method,
{\it Phys, Rev.} {\bf 35} (1930), 210--211.
\bibitem{SS} J. Sou\v cek, and V. Sou\v cek,
Morse-Sard theorem for real-analytic functions,
{\it Comment. Math. Univ. Carolinae} {\bf 13} (1972), 45--51.
\bibitem{St} C. A. Stuart,
Existence theory for the Hartree equation,
{\it Arch. Rational Mech. Anal.} {\bf 51} (1973), 60--69.
\bibitem{Wo} J. H. Wolkowisky,
Existence of solutions of the Hartree equations for N electrons an application of the Schauder-Tychonoff theorem,
{\it Indiana Univ. Math. J.} {\bf 22} (1972), 551--568.
\end{thebibliography}
\end{document}